\documentclass[11pt]{article}

\usepackage[margin=1in]{geometry}
\usepackage[utf8]{inputenc}
\usepackage[T1]{fontenc}
\usepackage{lmodern}
\usepackage{microtype}
\usepackage{amsmath,amssymb,amsthm,mathtools}
\usepackage{xcolor}
\usepackage[colorlinks=true,linkcolor=blue,citecolor=blue,urlcolor=blue]{hyperref}
\hypersetup{pdftitle={Parity and Pattern Detection in Permutation Streams},pdfauthor={Or Zamir}}
\usepackage{booktabs}
\usepackage{enumitem}

\newtheorem{theorem}{Theorem}[section]
\newtheorem{lemma}[theorem]{Lemma}
\newtheorem{proposition}[theorem]{Proposition}
\newtheorem{corollary}[theorem]{Corollary}
\theoremstyle{definition}

\theoremstyle{remark}

\usepackage{thmtools} 
\usepackage[nameinlink,noabbrev]{cleveref}

\crefname{theorem}{Theorem}{Theorems}
\Crefname{theorem}{Theorem}{Theorems}
\crefname{lemma}{Lemma}{Lemmas}
\Crefname{lemma}{Lemma}{Lemmas}
\crefname{proposition}{Proposition}{Propositions}
\Crefname{proposition}{Proposition}{Propositions}
\crefname{corollary}{Corollary}{Corollaries}
\Crefname{corollary}{Corollary}{Corollaries}

\newcommand{\E}{\mathop{\mathbb E}}
\newcommand{\D}{\mathbb D}
\newcommand{\Av}{\operatorname{Av}}
\newcommand{\inv}{\operatorname{inv}}
\newcommand{\sgn}{\operatorname{sgn}}
\newcommand{\supp}{\operatorname{supp}}
\newcommand{\pos}{\operatorname{pos}}
\newcommand{\ind}{\mathbf 1}

\title{Parity and Pattern Detection in Permutation Streams}
\author{
Mark Braverman \\ Princeton University
\and
Or Zamir \\ Tel Aviv University}
\date{}

\begin{document}
\maketitle

\begin{abstract}
Consider a permutation of $[n]$ whose values arrive one at a time. We resolve two questions about the space needed to decide natural properties of such input: First, computing the parity of the permutation requires $\Theta(n)$ bits, even with randomization and constant error, and a constant number of passes.
Second, every permutation pattern of length three can be detected deterministically in one pass using $O(\log n)$ bits. Together with the 2026 lower bounds of Berendsohn, this completes the classification of fixed permutation patterns; The optimal space complexity is $\Theta(\log n)$ for monotone patterns and patterns of length at most three, and $\Theta(n)$ for every other pattern.
As a consequence, we observe that we can verify BST traversals in streaming with logarithmic memory. 
\end{abstract}

\section{Introduction}
\label{sec:introduction}

Consider an input $\pi_1,\ldots,\pi_n$ that is promised to be a permutation of $[n]$. The values arrive in this order in the streaming model, and an algorithm must decide a property of the permutation while using as little memory space as possible. The promise that every value occurs exactly once can substantially change the space complexity of solving various problems. In particular, problems that are invariant to reordering (e.g. frequency moments) become trivial, and thus we study problems that depend only on the order itself.

Streaming questions about permutations already appeared in early work on distances and measures of sortedness. Cormode, Muthukrishnan, and Sahinalp~\cite{CormodeEtAl2001} used embeddings to obtain streaming algorithms for distances between permutations. Ajtai, Jayram, Kumar, and Sivakumar~\cite{AjtaiEtAl2002} studied approximate inversion counting, while Gopalan, Jayram, Krauthgamer, and Kumar~\cite{GopalanEtAl2007} studied the longest increasing subsequence and distance to monotonicity. The latter work also proved a linear-space lower bound for computing the longest increasing subsequence exactly, even when the input is a permutation.

Berendsohn~\cite{Berendsohn2026} studied this setting for permutation pattern detection: A permutation contains a pattern $\sigma$ if some subsequence has the same relative order as $\sigma$. For example, an occurrence of $312$ consists of positions $a<b<c$ with $\pi_a>\pi_c>\pi_b$. For each fixed-length monotone pattern, a standard increasing-subsequence algorithm uses $O(\log n)$ bits. 
For every fixed-length non-monotone pattern of length at least four, Berendsohn proved an $\Omega(n)$ lower bound. For the four remaining patterns,
\[
        132,\qquad 213,\qquad 231,\qquad 312,
\]
he showed algorithms using $O(\sqrt{n\log n})$ or $O(\sqrt n\log n)$ bits of space, but only an $\Omega(\log n)$ lower bound. Berendsohn conjectured that these patterns require $\widetilde\Theta(\sqrt n)$ space, which would imply a surprising trichotomy between three distinct complexity classes.

Another question left open by the above works concerns the parity of a permutation. Let
\[
        \inv(\pi)=\bigl|\{(i,j):i<j,\ \pi_i>\pi_j\}\bigr|,
        \qquad
        \sgn(\pi)=(-1)^{\inv(\pi)}.
\]
Counting inversions and approximating their number have been studied in the streaming setting~\cite{AjtaiEtAl2002}. Computing only their parity asks for a single bit, but the usual type of reductions based on communication complexity encounter a particular difficulty: when a permutation is split into two index sets, not necessarily contiguous, two players can compute its sign with one bit of communication (see e.g. the discussion in~\cite{Berendsohn2026}).

In this paper we resolve both questions. 
We show that permutation parity requires linear space, while all patterns of length three admit logarithmic-space algorithms.

\subsection{Permutation parity}

Our first result determines the one-pass space complexity of permutation parity, including randomized algorithms.

\begin{theorem}[Permutation parity]
\label{thm:parity-main}
Any randomized one-pass streaming algorithm that computes the sign of a permutation of $[n]$ with error at most $1/3$ on every input requires $\Omega(n)$ bits of space. 
More generally, a $p$-pass streaming algorithm requires $\Omega(n/p)$ bits. 
\end{theorem}

This bound is tight as there is a simple deterministic algorithm that uses $n+1$ bits and computes the parity of the input permutation: We store the set of values already seen as an $n$-bit vector. On receiving $x$, we add the parity of the number of marked values larger than $x$ to an accumulator bit, and then mark $x$. Each inversion contributes exactly once. 

We remark again that standard two-player communication complexity cannot result in a lower bound here:
Consider any partition of the index set into any two parts~$[n]=X_A\cup X_B$, and give the restriction of some permutation~$\pi \in S_n$ to~$X_A,X_B$ to two players respectively.
Both players can easily compute the partitions~$\left(X_A, X_B\right)$ and~$\left(\pi(X_A), \pi (X_B)\right)$ and in particular can compute the parity of the permutation~$\pi_0$ in which the values of~$\pi(X_A)$ appear in a sorted increasing order as the image of~$X_A$ and similarly~$\pi(X_B)$ appear in increasing order as the image of~$X_B$.
Now, we notice that the parity of~$\pi$ is simply the product of the parity of~$\pi_0$ with the two internal parities of the correct permutation needed to get from~$\pi_0$ to~$\pi$ on the elements of~$X_A$, a parity that~$A$ can compute alone, and the same for~$B$.
In particular, only one bit of communication is needed for the other party to correctly compute the parity of~$\pi$.

We thus prove the lower bound via a reduction from a communication complexity problem with \emph{three players}, circumventing the above impossibility. 
We then use a theorem from~\cite{BhangaleEtAl2025} to obtain a linear lower bound for that three-player communication problem.

\subsection{A dichotomy for pattern detection}

Our second result gives deterministic algorithms for all four unresolved patterns. Write $S_\sigma(n)$ for the minimum number of bits used by a one-pass algorithm that detects a fixed pattern $\sigma$ in permutations of $[n]$.

\begin{theorem}[Pattern dichotomy]
\label{thm:dichotomy}
For every fixed permutation pattern $\sigma$ of length at least two,
\[
        S_\sigma(n)=
        \begin{cases}
        \Theta_\sigma(\log n),
            & \text{if $\sigma$ is monotone or $|\sigma|=3$,}\\[2pt]
        \Theta_\sigma(n),
            & \text{otherwise.}
        \end{cases}
\]
\end{theorem}

The new algorithms reduce the space bound for each of the patterns $132,213,231,312$ left unresolved by Berendsohn to $O(\log n)$. In particular, we disprove Conjecture~6.1 in~\cite{Berendsohn2026}. The lower bounds in \Cref{thm:dichotomy} are due to Berendsohn. 

For completeness, we also close a small gap left in~\cite{Berendsohn2026} for the second case:
Their lower bound is~$\Omega_\sigma(n)$ while the naive upper bound of storing the entire input prefix requires~$O(n\log n)$ memory bits.
To obtain the exact linear bound in the second case, we use a theorem of Marcus and Tardos~\cite{MarcusTardos2004} which states that a fixed pattern has only exponentially many avoiding permutations. 
A prefix that has not yet contained the pattern can therefore be represented using $O(n)$ bits, including its set of values. Once the pattern appears, the algorithm can move to one accepting state. 

\subsection{Verifying binary search tree traversals}

The length-three results imply logarithmic-space algorithms for recognizing traversals of binary search trees.
In a preorder traversal, the root is visited first, followed by the left and right subtrees; in a postorder traversal, the left and right subtrees are visited before the root.
Given a permutation of $[n]$, we ask whether it is the preorder or postorder traversal of some binary search tree.
The tree itself is not supplied as part of the input.

These conditions have classical characterizations by forbidden patterns.
A permutation is a preorder traversal of a binary search tree if and only if it avoids $231$, and it is a postorder traversal if and only if it avoids $312$.
See Kozma~\cite[Lemma~1.4]{Kozma2016} for the preorder characterization and Levy and Tarjan~\cite[Lemma~1]{LevyTarjan2019} for both characterizations.
Our length-three algorithms therefore give the following consequence.

\begin{corollary}
Given a permutation of $[n]$ as a stream, one can decide deterministically in one pass using $O(\log n)$ bits whether it is the preorder traversal of some binary search tree.
The same bound holds for postorder traversals.
\end{corollary}

These permutation classes have also been studied in the analysis of binary search tree algorithms.
Chalermsook et al.~\cite{ChalermsookEtAl2015} used pattern avoidance to analyze the cost of the Greedy binary search tree algorithm, including its behavior on preorder sequences.
Levy and Tarjan~\cite{LevyTarjan2019} studied inserting preorder and postorder sequences into an initially empty tree using splaying.
Our result determines the space needed to recognize these traversal conditions from the sequence of keys alone.
We include proofs of the characterizations in \Cref{sec:bst}.

There is also an immediate consequence for stack sorting.
Knuth~\cite{Knuth1968} characterized the permutations sortable into increasing order using a single stack as precisely those avoiding $231$.
Thus our algorithm also recognizes such permutations deterministically in one pass using $O(\log n)$ bits.

Our traversal algorithms also relate to the study of checking data structures in the streaming model. Chu, Kannan, and McGregor~\cite{ChuEtAl2007}, and later Fran\c{c}ois and Magniez~\cite{FrancoisMagniez2013}, studied whether a stream of priority-queue operations is consistent with the required semantics. Magniez, Mathieu, and Nayak~\cite{MagniezEtAl2014} studied the recognition of well-parenthesized expressions, a related problem of checking nested structure with limited memory. Here the input is a sequence of keys, promised to contain each element of $[n]$ exactly once, and we ask whether it describes a valid binary search tree traversal. This promise allows the traversal condition to be checked using a constant number of integer accumulators, each of logarithmic size.

\paragraph{Acknowledgments and AI Usage.}
The authors thank Berendsohn and Kozma for bringing these problems to our attention.
The parity lower bound was derived by the authors, and the length-three pattern upper bounds were mostly derived by ChatGPT. 
AI models were also used in the writing of this paper, albeit the final text is edited and verified by the authors.

\paragraph{Organization.}
\Cref{sec:overview} gives an overview of the arguments, and \Cref{sec:preliminaries} defines the notation. \Cref{sec:parity} proves the parity bounds. \Cref{sec:length-three} gives the length-three algorithms and their interpretation as traversal verification. \Cref{sec:classification} proves the linear upper bound for fixed patterns and completes the dichotomy.

\section{Overview}
\label{sec:overview}

\paragraph{Three-party communication for parity.}
Suppose Alice, Bob, and Charlie receive $x,y,z\in[3]^m$, respectively, with the promise that $(x_i,y_i,z_i)$ is a permutation of $[3]$ for every $i$. 

We map the three permutation values of coordinate $i$ to the interval $\{3i-2,3i-1,3i\}$. Alice streams her values in increasing coordinate order, followed by Bob and then Charlie. This gives a permutation of $[3m]$ consisting of three increasing blocks.
We observe that the sign of the resulting permutation, up to a fixed global correction that depends only on~$m$, is the product of the $m$ local signs of the independent size-three permutations.

It remains to show an appropriate XOR lemma that would imply that computing this product requires $\Omega(m)$ communication between the three players. 
A theorem of Bhangale et al.~\cite{BhangaleEtAl2025} gives exactly the needed estimate: the product of the local signs has exponentially small correlation with any product $F(x)G(y)H(z)$ of bounded functions. Taking $F,G,H$ to be indicator functions bounds the discrepancy of every communication rectangle, and summing over protocol transcripts gives the lower bound.

\paragraph{Weighted inversions for length-three patterns.}
Despite the ordinary inversion count being difficult to compute in small space, we observe that a weighted version is much easier to compute. 
Define
\[
        W(\pi)=\sum_{\substack{i<j\\\pi_i>\pi_j}}(\pi_i-\pi_j).
\]
For a permutation of $[n]$, a direct calculation gives
\[
        W(\pi)=\sum_{v=1}^n v^2-\sum_{i=1}^n i\pi_i.
\]
Thus $W$ can be computed with logarithmic space. We compare it with two other quantities that can also be computed in logarithmic space. Each comparison is arranged so that the difference is nonnegative and is positive exactly when the forbidden pattern occurs.

For $312$, let $M_i$ be the maximum value in the prefix ending at position $i$. Consider
\[
        \sum_{i=1}^n\sum_{v=\pi_i+1}^{M_i}(v-\pi_i).
\]
Every value $v$ in an inner sum appears either before or after $\pi_i$. The earlier values contribute exactly $W(\pi)$. Each later value forms a $312$ occurrence with $\pi_i$ and the earlier maximum. Subtracting $W$ leaves a sum of non-negative weights, and each positive one corresponds to an occurrence of the pattern.
For $231$, we present a similar summation.
The other two patterns are equivalent by symmetries. 

\paragraph{Linear space from counting avoiding prefixes.}
Fix a pattern $\sigma$, and let $a_t$ be the number of permutations of $[t]$ that avoid it. A length-$t$ avoiding prefix over $[n]$ is determined by its set of $t$ values and one of these $a_t$ relative orders. The theorem of Marcus and Tardos~\cite{MarcusTardos2004} gives $a_t\le C_\sigma^t$. Hence the total number of avoiding prefixes is at most
\[
        \sum_{t=0}^n\binom nt C_\sigma^t=(1+C_\sigma)^n.
\]
Using one state for each such prefix and one accepting state gives the $O_\sigma(n)$-bit upper bound.

\section{Preliminaries}
\label{sec:preliminaries}

All logarithms are in base two. We write $[n]=\{1,\ldots,n\}$ and $S_n$ for the set of permutations of $[n]$. A sequence $v_1,\ldots,v_k$ of distinct values is \emph{order-isomorphic} to $\sigma\in S_k$ if $v_a<v_b$ exactly when $\sigma_a<\sigma_b$. A permutation \emph{contains} $\sigma$ if it has an order-isomorphic subsequence, and otherwise \emph{avoids} $\sigma$. We write $\Av_t(\sigma)$ for the set of permutations of $[t]$ that avoid $\sigma$, including the empty permutation when $t=0$. A pattern is \emph{monotone} if it is increasing or decreasing.

For a permutation $\pi$, let $\pos_\pi(v)$ be the position of value $v$, omitting the subscript when the permutation is clear. We use the two elementary sums
\[
        T(d)=\frac{d(d+1)}2\quad(d\ge0),
        \qquad
        Q(n)=\sum_{v=1}^n v^2=\frac{n(n+1)(2n+1)}6.
\]

\section{The Parity of a Permutation}
\label{sec:parity}

We first give the simple upper bound. We then prove a communication lower bound for the product of independent signs in $S_3$ and reduce that problem to permutation parity.

\subsection{A simple \texorpdfstring{$n+1$}{n+1}-bit upper bound}

\begin{proposition}
\label{prop:parity-upper}
The sign of a permutation of $[n]$ can be computed deterministically in one pass using $n+1$ bits of space.
\end{proposition}

\begin{proof}
Maintain a vector $b\in\{0,1\}^n$ and a bit $r$, both initially zero. On receiving a value $x$, perform the transition
\[
        r\leftarrow r\oplus\bigoplus_{y=x+1}^n b_y,
        \qquad
        b_x\leftarrow1,
\]
where $\oplus$ denotes addition modulo two. The vector records exactly which values have appeared. The inner sum therefore counts, modulo two, the inversions whose second position is the current position. Every inversion is counted once, so the final answer is $(-1)^r$.
\end{proof}

\subsection{XOR Lemma for Permutation Signs}
\label{sec:analytic}

Let $\D=\{z\in\mathbb C:|z|\le1\}$. If $\mu$ is a distribution on a product of three finite sets $\Sigma\times\Gamma\times\Phi$, its support graph on $\Sigma$ and $\Gamma$ has an edge $(x,y)$ whenever $\mu(x,y,z)>0$ for some $z$. Define the other two support graphs in the same way. We call $\mu$ \emph{pairwise connected} when all three graphs are connected, after removing symbols of zero marginal probability.

A function $T:\supp(\mu)\to\D$ is \emph{nonembeddable over $\D$} if there are no functions $a:\Sigma\to\D$, $b:\Gamma\to\D$, and $c:\Phi\to\D$ satisfying
\[
        a(x)b(y)c(z)T(x,y,z)=1
        \qquad\text{for every }(x,y,z)\in\supp(\mu).
\]
We use the following result of Bhangale, Braverman, Khot, Liu, and Minzer. It is Theorem~2 in the full version of their paper~\cite{BhangaleEtAl2025}.

\begin{theorem}[Bhangale, Braverman, Khot, Liu, and Minzer]
\label{thm:analytic}
Suppose that $\mu$ is pairwise connected and that $T:\supp(\mu)\to\D$ is nonembeddable over $\D$. There is a constant $\alpha=\alpha(\mu,T)>0$ such that for every positive integer $m$ and every $F:\Sigma^m\to\D$, $G:\Gamma^m\to\D$, and $H:\Phi^m\to\D$,
\[
        \left|
        \E_{(x,y,z)\sim\mu^{\otimes m}}
        \left[F(x)G(y)H(z)\prod_{i=1}^m T(x_i,y_i,z_i)\right]
        \right|
        \le 2^{-\alpha m}.
\]
\end{theorem}

We apply this theorem to $\Sigma=\Gamma=\Phi=[3]$ and the uniform distribution $\mu$ on the six permutations of $(1,2,3)$. On its support define
\[
        \varepsilon(a,b,c)=\sgn(a,b,c).
\]
For vectors $x,y,z\in[3]^m$ whose coordinate triples lie in this support, write
\[
        f_m(x,y,z)=\prod_{i=1}^m\varepsilon(x_i,y_i,z_i).
\]
Equivalently, $f_m$ computes the XOR of the $m$ coordinate parities.

\begin{lemma}
\label{lem:sign-correlation}
There is an absolute constant $\alpha>0$ such that, for all $m\ge1$ and all $F,G,H:[3]^m\to\D$,
\[
        \left|\E_{\mu^{\otimes m}}
        \bigl[f_m(x,y,z)F(x)G(y)H(z)\bigr]\right|
        \le2^{-\alpha m}.
\]
\end{lemma}

\begin{proof}
In each pairwise support graph, two symbols are adjacent exactly when they differ. The graph is $K_{3,3}$ with a perfect matching removed, hence a connected cycle of length six.

Suppose that $a,b,c:[3]\to\D$ satisfy
\[
        a(u)b(v)c(w)\varepsilon(u,v,w)=1
\]
for every permutation $(u,v,w)$ of $[3]$. Put
\[
        P=\prod_{j=1}^3a(j)b(j)c(j).
\]
Multiplying the equations for the even permutations $123,231,312$ gives $P=1$, since each symbol occurs once in each coordinate. Multiplying the equations for the odd permutations $132,213,321$ instead gives $-P=1$. This is a contradiction. Thus $\varepsilon$ is nonembeddable over $\D$, and \Cref{thm:analytic} applies.
\end{proof}

The main parallel-repetition theorem in Bhangale et al. has a different hypothesis, excluding nonconstant additive embeddings of the question distribution into $\mathbb Z$. Our distribution does have such an embedding: the three maps $j\mapsto j-2$ sum to zero on every permutation of $[3]$. \Cref{thm:analytic} requires nonembeddability of the target over $\D$, which we have just verified, and does not impose that additional hypothesis on the distribution.

Consider the three-party number-in-hand communication problem for $f_m$. Alice knows $x\in[3]^m$, Bob knows $y\in[3]^m$, and Charlie knows $z\in[3]^m$. They are promised that each $(x_i,y_i,z_i)$ is a permutation of $[3]$. The players communicate by writing bits on a common blackboard, and may use public randomness. There is no restriction on the number of rounds or on their local computation. We count the final output bit as part of the communication.

\begin{lemma}
\label{lem:communication}
Every randomized protocol computing $f_m$ with error at most $\delta<1/2$ on each promised input communicates at least
\[
        \alpha m+\log(1-2\delta)
\]
bits in the worst case, where $\alpha$ is the constant in \Cref{lem:sign-correlation}.
\end{lemma}

\begin{proof}
A rectangle is a set $R=A\times B\times C$ with $A,B,C\subseteq[3]^m$. Taking $F,G,H$ to be the indicator functions of $A,B,C$ in \Cref{lem:sign-correlation} gives
\[
        \left|\E_{\mu^{\otimes m}}[f_m\ind_R]\right|
        \le2^{-\alpha m}.
\]
The distribution is supported on promised inputs, so no condition on the other inputs of $R$ is needed.

Fix a deterministic protocol communicating at most $L$ bits, including its output. Each complete transcript determines a rectangle: once the transcript is fixed, consistency with all messages sent by a particular player is a condition only on that player's input. On promised inputs these rectangles partition the domain. Their number is at most $2^L$, since the complete transcripts are leaves of a binary protocol tree of depth at most $L$.

Let $R_1,\ldots,R_t$ be the rectangles and $o_1,\ldots,o_t\in\{-1,1\}$ their outputs. If $\Pi$ denotes the protocol's output, then
\[
        \left|\E[f_m\Pi]\right|
        =\left|\sum_{j=1}^t o_j\E[f_m\ind_{R_j}]\right|
        \le t\,2^{-\alpha m}
        \le2^{L-\alpha m}.
\]
For a randomized protocol, condition on all random coins and then average this inequality. A worst-case error bound of $\delta$ gives
\[
        1-2\delta
        \le\E_{(x,y,z),\,\mathrm{coins}}[f_m(x,y,z)\Pi(x,y,z)]
        \le2^{L-\alpha m}.
\]
Taking logarithms proves the result.
\end{proof}

In particular, error at most $1/3$ requires at least $\alpha m-\log3$ communicated bits. The argument allows any number of rounds, which will also give the lower bound for several streaming passes.

\subsection{Reduction from the communication problem to permutation parity}
\label{sec:reduction}

Given a promised input $(x,y,z)$ to $f_m$, define
\begin{align*}
        A(x)&=(3(i-1)+x_i)_{i=1}^m,\\
        B(y)&=(3(i-1)+y_i)_{i=1}^m,\\
        C(z)&=(3(i-1)+z_i)_{i=1}^m,
\end{align*}
and let $\pi=A(x)B(y)C(z)$ be their concatenation. Each player can generate its own block without knowing the other inputs. Every coordinate $i$ contributes exactly the values $3i-2,3i-1,3i$, so $\pi$ is a permutation of $[3m]$. Each of the three blocks is increasing.

\begin{lemma}
\label{lem:reduction-sign}
The permutation constructed above satisfies
\[
        \inv(\pi)=3\binom m2+
        \sum_{i=1}^m\inv(x_i,y_i,z_i).
\]
Consequently,
\[
        \sgn(\pi)=(-1)^{3\binom m2}f_m(x,y,z).
\]
\end{lemma}

\begin{proof}
Within a coordinate $i$, the three values occur in the order specified by $(x_i,y_i,z_i)$, giving the corresponding local inversion count.

Fix two different coordinates $i<j$. All three values of coordinate $i$ are smaller than all three values of coordinate $j$. Within one block, coordinate $i$ precedes coordinate $j$, so there is no inversion between them. Between different blocks, an inversion occurs precisely when the value of coordinate $j$ lies in the earlier block. There are three such pairs of blocks: $A$ and $B$, $A$ and $C$, and $B$ and $C$. Thus this pair of coordinates contributes exactly three inversions. Summing over coordinates proves the first identity, and taking parity proves the second.
\end{proof}

\begin{proof}[Proof of \Cref{thm:parity-main}]
Suppose first that $n=3m$, and let an $s$-bit streaming algorithm compute the sign with error at most $1/3$. Alice runs the algorithm on $A(x)$ and sends its state to Bob. Bob continues on $B(y)$ and sends the resulting state to Charlie. Charlie processes $C(z)$ and multiplies the answer by the fixed correction $(-1)^{3\binom m2}$. Revealing this answer uses one further bit. With public coins the players reproduce the random choices of the streaming algorithm, so the error is unchanged.

The protocol communicates at most $2s+1$ bits. By \Cref{lem:communication},
\[
        2s+1\ge\alpha m-\log3,
\]
and hence $s=\Omega(m)=\Omega(n)$.

For $p$ passes, the players repeat the simulation, with Charlie returning the memory to Alice between passes. There are $2p+(p-1)=3p-1$ state transmissions, followed by one output bit. Therefore
\[
        (3p-1)s+1\ge\alpha m-\log3,
\]
which gives $s=\Omega(n/p)$ for all sufficiently large $n$.

For arbitrary $n$, put $m=\lfloor n/3\rfloor$ and append the remaining largest values, in increasing order, to Charlie's block. These values create no additional inversions. The input still has three increasing blocks of fixed known lengths, and the same argument applies. Finally, \Cref{prop:parity-upper} proves the one-pass upper bound.
\end{proof}

For completeness, we remark that there is also a simple matching upper bound that improves with the number of passes.

\begin{proposition}
\label{prop:parity-passes}
For $1\le p\le n$, permutation parity can be computed deterministically in $p$ forward passes using $O(\lceil n/p\rceil+\log n)$ bits.
\end{proposition}

\begin{proof}
Partition $[n]$ into $p$ consecutive intervals, each of size at most $\lceil n/p\rceil$. In the pass assigned to an interval $I=[\ell,u]$, count only inversions whose smaller value belongs to $I$.

During this pass, maintain a bitmap of the values of $I$ already seen, and a bit $h$ recording the parity of the number of values above $u$ already seen. On receiving $x>u$, toggle $h$. On receiving $x\in I$, add to a global answer bit the XOR of $h$ and the parity of the marked values in $I$ larger than $x$, and then mark $x$. Values below $\ell$ require no update. Reset the bitmap and $h$ between passes, but retain the global answer bit.

When $x\in I$, the update counts exactly the earlier values larger than $x$. Each inversion is counted in the unique pass containing its smaller value. The bitmap requires at most $\lceil n/p\rceil$ bits; the interval boundaries, counters, and answer require $O(\log n)$ further bits.
\end{proof}

\section{Detecting Patterns of Length Three}
\label{sec:length-three}

We give algorithms for $312$ and $231$. Complementation then gives $132$ and $213$, respectively. The two monotone patterns follow from the usual increasing-subsequence algorithm, which we recall in \Cref{sec:classification} and was also cited in~\cite{Berendsohn2026}.

Our algorithms use several known identities on permutations or close variants of them~\cite{Kobayashi2011,AyyerEtAl2011,SackUlfarsson2012,BarilEtAl2018}.

\subsection{A weighted inversion identity}

The following identity appeared in \cite{Kobayashi2011}. We include a proof for completeness and to match our notation.
\begin{lemma}
\label{lem:weighted-inversions}
For every permutation $\pi\in S_n$,
\[
        W(\pi):=\sum_{\substack{i<j\\\pi_i>\pi_j}}(\pi_i-\pi_j)
        =Q(n)-\sum_{i=1}^n i\pi_i.
\]
In particular, $W(\pi)$ can be computed deterministically in one pass with $O(\log n)$ bits.
\end{lemma}

\begin{proof}
Fix a value $x=\pi_i$, and let $\ell$ be the number of earlier values smaller than $x$. There are $x-1-\ell$ smaller values after $x$, and $i-1-\ell$ larger values before $x$. In the expansion of the weighted inversion sum, the coefficient of $x$ is therefore
\[
        (x-1-\ell)-(i-1-\ell)=x-i.
\]
Summing over values gives
\[
        W(\pi)=\sum_{i=1}^n\pi_i(\pi_i-i)
        =\sum_{v=1}^n v^2-\sum_{i=1}^n i\pi_i.
\]
Our streaming algorithm maintains the current position and the sum $\sum_i i\pi_i$. Both require $O(\log n)$ bits, since the sum is at most $n^3$. The quantity $Q(n)$ depends only on $n$.
\end{proof}

\subsection{The pattern \texorpdfstring{$312$}{312}}

Let $M_i=\max\{\pi_1,\ldots,\pi_i\}$ and define
\[
        E_{312}(\pi)=\sum_{i=1}^n T(M_i-\pi_i)-W(\pi).
\]
The next identity is the reason that testing this single integer suffices.

\begin{lemma}
\label{lem:312-identity}
For every $\pi\in S_n$,
\[
        E_{312}(\pi)
        =\sum_{\substack{i<j\\\pi_i<\pi_j<M_i}}(\pi_j-\pi_i).
\]
Consequently, $E_{312}(\pi)\ge0$, with equality if and only if $\pi$ avoids $312$.
\end{lemma}

\begin{proof}
Expand each triangular number as
\[
        T(M_i-\pi_i)=\sum_{v=\pi_i+1}^{M_i}(v-\pi_i).
\]
Every value in this interval occurs exactly once in $\pi$. Split the sum according to whether $v$ occurs before or after position $i$. Every earlier value larger than $\pi_i$ is at most $M_i$, so the earlier terms, summed over $i$, are exactly the weighted inversions $W(\pi)$. The remaining terms are
\[
        \sum_{\substack{i<j\\\pi_i<\pi_j\le M_i}}(\pi_j-\pi_i).
\]
Since $M_i$ already occurs in the prefix through $i$, a later value cannot equal $M_i$. This gives the stated identity.

Every summand is positive. If the pair $(i,j)$ occurs in the sum, choose a position $a\le i$ attaining $M_i$. The inequalities $M_i>\pi_j>\pi_i$ imply $a<i<j$, so $(a,i,j)$ is a $312$ occurrence. Conversely, if $a<b<c$ and $\pi_a>\pi_c>\pi_b$, then
\[
        M_b\ge\pi_a>\pi_c>\pi_b,
\]
and the pair $(b,c)$ contributes a positive term. Thus the sum vanishes exactly on $312$-avoiding permutations.
\end{proof}

\begin{proposition}
\label{prop:312-algorithm}
The pattern $312$ can be detected deterministically in one pass with $O(\log n)$ bits and $O(n)$ arithmetic operations.
\end{proposition}

\begin{proof}
Initialize three registers by
\[
        i=0,\qquad M=0,\qquad E=-Q(n).
\]
On receiving $x$, update them in the following order:
\[
        i\leftarrow i+1,\qquad
        M\leftarrow\max\{M,x\},\qquad
        E\leftarrow E+ix+T(M-x).
\]
At the end, output that the pattern occurs if and only if $E>0$. By \Cref{lem:weighted-inversions}, the final value is
\[
        -Q(n)+\sum_{i=1}^n\bigl(i\pi_i+T(M_i-\pi_i)\bigr)
        =E_{312}(\pi),
\]
so correctness follows from \Cref{lem:312-identity}.

The registers $i$ and $M$ are at most $n$. Each term added to $E$ has magnitude $O(n^2)$, so $E$ and every intermediate expression have magnitude $O(n^3)$. A constant number of $O(\log n)$-bit registers therefore suffices, with a constant number of arithmetic operations per element.
\end{proof}

\subsection{The pattern \texorpdfstring{$231$}{231}}

Partition $\pi$ into its maximal contiguous decreasing runs
\[
        R_1,R_2,\ldots,R_k.
\]
Let $b_r$ be the last, hence smallest, value in $R_r$. For a value $h$ in the permutation, write $b(h)$ for the bottom of its run. Define
\[
        U(\pi)=\sum_{r=1}^k\sum_{h\in R_r}T(h-b_r)
        =\sum_{h=1}^n T(h-b(h)).
\]

\begin{lemma}
\label{lem:231-characterization}
A permutation $\pi$ avoids $231$ if and only if
\[
        b_1<b_2<\cdots<b_k
        \qquad\text{and}\qquad
        U(\pi)=W(\pi).
\]
\end{lemma}

\begin{proof}
Suppose first that $b_r>b_{r+1}$ for two consecutive runs, and let $h$ be the first value of $R_{r+1}$. Maximality of the runs gives $b_r<h$. In particular, $h$ and $b_{r+1}$ are distinct positions, and the three values
\[
        b_r,\ h,\ b_{r+1}
\]
occur in this order and form $231$. Thus increasing run bottoms are necessary for avoidance.

Assume from now on that all run bottoms increase. If $v<h$ occurs after $h$, then $v\ge b(h)$. This is immediate when the two values are in the same run. If $v$ belongs to a later run, its run bottom is larger than $b(h)$, so the inequality still holds. Therefore every inversion whose larger value is $h$ is represented in the expansion
\[
        T(h-b(h))=\sum_{v=b(h)}^{h-1}(h-v).
\]
Each integer in this interval occurs somewhere in the permutation. Subtracting the contributions of those occurring after $h$ gives
\[
        U(\pi)-W(\pi)
        =\sum_{h=1}^n
        \sum_{\substack{b(h)\le v<h\\\pos(v)<\pos(h)}}(h-v)
        \ge0.
\]

Every summand gives a $231$ occurrence. Indeed, the bottom $b(h)$ cannot occur before $h$. If $h=b(h)$, the inner interval is empty; otherwise the bottom occurs strictly after $h$. Hence an earlier value $v$ in this sum satisfies $b(h)<v<h$, and the three values $v,h,b(h)$ occur in that order.

Conversely, suppose $v,h,z$ is a $231$ occurrence. Since $z<h$ occurs after $h$, the increasing-bottom condition gives
\[
        b(h)\le z<v<h.
\]
The earlier value $v$ therefore contributes the positive term $h-v$ to the inner sum for $h$. Consequently, under the increasing-bottom condition, $U-W$ is zero exactly when $\pi$ avoids $231$. Together with the first part of the proof, this establishes the characterization.
\end{proof}

The quantity $U$ can be computed without storing the runs. If a decreasing run $R$ has bottom $b$, length $c$, sum $s$, and squared sum $q$, then
\begin{align*}
        \sum_{h\in R}T(h-b)
        &=\frac12\sum_{h\in R}\bigl((h-b)^2+(h-b)\bigr)\\
        &=\frac{q+(1-2b)s+c(b^2-b)}2.
\end{align*}
This is an integer, since each term $(h-b)(h-b+1)$ is even.

\begin{proposition}
\label{prop:231-algorithm}
The pattern $231$ can be detected deterministically in one pass with $O(\log n)$ bits and $O(n)$ arithmetic operations.
\end{proposition}

\begin{proof}
Maintain a position counter $i$, a sum $D=\sum_{j\le i}j\pi_j$, an accumulator $U$ for completed runs, and the length, sum, squared sum, and last value of the current run. Also maintain the bottom $b_{\mathrm{prev}}$ of the previous completed run, initially zero. All other counters and sums are initially zero.

On receiving $x$, increment $i$ and add $ix$ to $D$. If the current run is nonempty and $x$ is larger than its last value, finalize that run before inserting $x$. To finalize a run, let $b$ be its last value and $c,s,q$ its three statistics. If $b\le b_{\mathrm{prev}}$, report that $231$ occurs and move to an absorbing accepting state. Otherwise, add
\[
        \frac{q+(1-2b)s+c(b^2-b)}2
\]
to $U$, set $b_{\mathrm{prev}}=b$, and reset the three run statistics to zero. After any necessary finalization, insert $x$ into the current run by adding $1,x,x^2$ to its length, sum, and squared sum, and recording $x$ as its last value.

After the final element, finalize the last run in the same way. If the algorithm has not already accepted, its completed bottoms increase and its accumulator is $U(\pi)$. Report that the pattern occurs exactly when
\[
        U\ne Q(n)-D.
\]
Correctness follows from \Cref{lem:weighted-inversions,lem:231-characterization}.

There are a constant number of registers. The run length is at most $n$, its sum is at most $n^2$, and its squared sum is at most $n^3$. The accumulators $D,U$ and every term in the run formula have magnitude $O(n^3)$. Thus all storage, including intermediate arithmetic, is $O(\log n)$ bits. Each element is inserted once, and at most one run is finalized per element, giving $O(n)$ arithmetic operations.
\end{proof}

\subsection{Complementation}

The \emph{complement} of $\pi\in S_n$ is the permutation $\overline\pi$ with
\[
        \overline\pi_i=n+1-\pi_i.
\]
Complementation reverses all comparisons without changing the order of positions. It maps an occurrence of a pattern $\sigma\in S_k$ to an occurrence of the pattern with values $k+1-\sigma_i$. In particular, it exchanges $312$ with $132$, and $231$ with $213$. The transformation is performed on each input element as it arrives.

\begin{corollary}
\label{cor:four-patterns}
Each pattern in $\{132,213,231,312\}$ can be detected deterministically in one pass using $O(\log n)$ bits and $O(n)$ arithmetic operations.
\end{corollary}

We remark that the proofs depend on more than the distinctness of the input values. In \Cref{lem:312-identity}, every integer between $\pi_i$ and $M_i$ must occur either before or after position $i$. The proof for $231$ uses the analogous fact for the interval from a run bottom to a value in that run. Missing values would leave additional positive terms that need not correspond to any pattern occurrence. This is consistent with Berendsohn's linear lower bound for all four patterns when the input is an arbitrary sequence of distinct values~\cite[Theorem~1.3]{Berendsohn2026}.

\subsection{Binary search tree traversals}
\label{sec:bst}

A binary search tree on $[n]$ has each label once, and all labels in the left subtree of a node are smaller than its label, while all labels in the right subtree are larger. We use the conventions root-left-right for preorder and left-right-root for postorder. The following elementary characterizations connect these traversal conditions to our algorithms.

\begin{lemma}
\label{lem:bst-characterization}
A permutation of $[n]$ is the preorder traversal of a binary search tree if and only if it avoids $231$. It is the postorder traversal of a binary search tree if and only if it avoids $312$.
\end{lemma}

\begin{proof}
For preorder, let $r$ be the first value. If the permutation avoids $231$, every later value smaller than $r$ must precede every later value larger than $r$. Otherwise $r,h,\ell$, with $h>r>\ell$, would be a $231$ occurrence. Thus the permutation has the form $rLR$, with the values of $L$ smaller than $r$ and those of $R$ larger. Both subsequences avoid $231$, so induction constructs the two subtrees.

Conversely, suppose $rLR$ is a valid preorder traversal. By induction, $L$ and $R$ avoid $231$. An occurrence using the root would need a value of $R$ followed by one of $L$, which is impossible. An occurrence crossing from $L$ to $R$ is also impossible, since all values in the later block are larger than all values in the earlier block, whereas the final value of $231$ is its smallest. Hence the whole traversal avoids $231$.

For postorder, let $r$ be the last value. Avoidance of $312$ again forces every value below $r$ to precede every value above $r$: a larger value $h$ before a smaller value $\ell$ would give the occurrence $h,\ell,r$. The resulting decomposition $LRr$ yields a tree by induction. Conversely, an occurrence of $312$ using the root would require a value of $R$ before one of $L$. An occurrence crossing the two subtrees cannot have its largest value first. Induction inside each subtree completes the proof.
\end{proof}

\begin{corollary}
\label{cor:bst}
Given a promised permutation of $[n]$, verifying whether it is a preorder traversal, or whether it is a postorder traversal, of some binary search tree can be done deterministically in one pass using $O(\log n)$ bits and $O(n)$ arithmetic operations. Both problems have deterministic space complexity $\Theta(\log n)$.
\end{corollary}

\begin{proof}
Apply \Cref{prop:231-algorithm,prop:312-algorithm} and \Cref{lem:bst-characterization}. The lower bounds are the logarithmic pattern-detection lower bounds stated in \Cref{thm:berendsohn} below; complementing an output bit does not change the state count.
\end{proof}

\section{A Dichotomy for Fixed Patterns}
\label{sec:classification}

We complete the proof of \Cref{thm:dichotomy}. The logarithmic algorithms for the remaining patterns were proved in \Cref{sec:length-three}. We recall the monotone algorithm, give the general linear upper bound, and state the lower bounds from Berendsohn that complete the classification.

The standard increasing-subsequence algorithm, originating with Schensted~\cite{Schensted1961} and used in~\cite{Berendsohn2026}, immediately gives the following bound. 

\begin{lemma}
\label{lem:monotone}
The increasing and decreasing patterns of length $k$ can be detected deterministically in one pass with $O(k\log n)$ bits. For fixed $k$, the algorithm uses $O(n)$ arithmetic operations.
\end{lemma}

\subsection{A linear upper bound for every fixed pattern}

A theorem of Marcus and Tardos implies the following enumerative bound~\cite[Corollary~2]{MarcusTardos2004}. 

\begin{theorem}[Marcus and Tardos]
\label{thm:stanley-wilf}
For every fixed permutation pattern $\sigma$, there is a constant $C_\sigma\ge1$ such that
\[
        |\Av_t(\sigma)|\le C_\sigma^t
        \qquad\text{for every integer }t\ge0.
\]
\end{theorem}

\begin{lemma}
\label{lem:linear-upper}
For every fixed pattern $\sigma$, there is a deterministic one-pass streaming algorithm for detecting $\sigma$ in permutations of $[n]$ with at most
\[
        1+(1+C_\sigma)^n
\]
states. In particular, it uses $O_\sigma(n)$ bits.
\end{lemma}

\begin{proof}
For each sequence $w$ of distinct elements of $[n]$ that avoids $\sigma$, create a state $q_w$. Include the empty sequence as the initial state. Add one further state $q_{\mathrm{yes}}$, which is accepting and absorbing.

On input $x$ in state $q_w$, if $x$ has already appeared in $w$, choose an arbitrary transition, since such an input violates the promise. Otherwise consider the concatenation $wx$. If it contains $\sigma$, move to $q_{\mathrm{yes}}$; if it avoids $\sigma$, move to $q_{wx}$. Once a prefix contains the pattern, every extension contains it, so these transitions recognize the desired property.

A length-$t$ avoiding sequence is uniquely determined by its set of $t$ values and its relative order in $\Av_t(\sigma)$. The total number of nonaccepting states is therefore
\[
        \sum_{t=0}^n\binom nt|\Av_t(\sigma)|
        \le\sum_{t=0}^n\binom nt C_\sigma^t
        =(1+C_\sigma)^n.
\]
Adding the accepting state and taking the logarithm gives the claimed bound.
\end{proof}

\subsection{Completing the classification}

We use the following previously established lower bounds. The one-pass statements are Theorems~1.1 and~1.4 of Berendsohn~\cite{Berendsohn2026}; the extension of the second statement to several passes is his Theorem~4.4 and the randomized extension discussed there.

\begin{theorem}[Berendsohn]
\label{thm:berendsohn}
For every fixed pattern $\sigma$ of length at least two, one-pass detection in permutations of $[n]$ requires $\Omega_\sigma(\log n)$ bits.
If $\sigma$ is non-monotone and has length at least four, detection requires $\Omega_\sigma(n)$ bits. 
\end{theorem}

\begin{proof}[Proof of \Cref{thm:dichotomy}]
If $\sigma$ is monotone, \Cref{lem:monotone} gives $O_\sigma(\log n)$ space. If $|\sigma|=3$, the same bound follows either from that lemma or from \Cref{cor:four-patterns}. In both cases the logarithmic lower bound is supplied by \Cref{thm:berendsohn}.

Every other pattern of length at least two is non-monotone and has length at least four. For these patterns, \Cref{thm:berendsohn} gives $\Omega_\sigma(n)$ space and \Cref{lem:linear-upper} gives $O_\sigma(n)$ space. The bounds for length three follow from \Cref{prop:312-algorithm,prop:231-algorithm,lem:monotone} and complementation.
\end{proof}

\bibliography{permutation_streaming}
\bibliographystyle{alpha}

\end{document}